\documentclass[11pt]{article}
\usepackage{amsmath}
\usepackage{amsfonts}
\usepackage{amssymb}
\usepackage{amsthm}
\usepackage{thmtools, thm-restate}
\usepackage{fullpage}
\usepackage{comment}
\usepackage{graphicx}
\usepackage{hyperref}
\usepackage{lmodern}
\usepackage{caption}
\usepackage{lipsum}
\usepackage{graphicx}
\usepackage{subcaption}
\usepackage{appendix}
\usepackage[normalem]{ulem}

\usepackage[
backend=biber,
style=alphabetic,
sorting=nyt,
maxbibnames=99 
]{biblatex}

\newtheorem{lemma}{Lemma}
\newtheorem{theorem}{Theorem}
\newtheorem{claim}{Claim}
\newtheorem{corollary}{Corollary}
\newtheorem{definition}{Definition}

\newtheorem{proposition}{Proposition} 

\newcommand{\dist}{\mathrm{dist}}
\newcommand{\rank}{\mathrm{rank}}
\newcommand{\scsize}{\mathrm{SC}}
\newcommand{\cond}{|}

\providecommand{\keywords}[1]
{
  \textbf{\textit{Keywords---}} #1
}

\title{Secret Sharing on Superconcentrator}
\date{}
\author{}
\author{
Yuan Li\footnote{Fudan University, China. Email: yuan\_li@fudan.edu.cn}}
\date{March 2022}

\begin{document}

\maketitle

\abstract{
We study the arithmetic circuit complexity of threshold secret sharing schemes by characterizing the graph-theoretic properties of arithmetic circuits that compute the shares. Using information inequalities, we prove that any unrestricted arithmetic circuit (with arbitrary gates and unbounded fan-in) computing the shares must satisfy superconcentrator-like connectivity properties. Specifically, when the inputs consist of the secret and $t-1$ random elements, and the outputs are the $n$ shares of a $(t, n)$-threshold secret sharing scheme, the circuit graph must be a $(t, n)$-concentrator; moreover, after removing the secret input, the remaining graph is a $(t-1, n)$-concentrator. Conversely, we show that any graph satisfying these properties can be transformed into a linear arithmetic circuit computing the shares of a threshold secret sharing scheme, assuming a sufficiently large field. As a consequence, we derive upper and lower bounds on the arithmetic circuit complexity of computing the shares in threshold secret sharing schemes.
}


\section{Introduction}

Understanding the arithmetic complexity of \emph{secret sharing} (hereafter referred to as ``SS'') is an important problem in theoretical computer science.

For a general access structure, the complexity of a secret-sharing scheme is typically measured by the ratio of the total share size to the secret size. Csirmaz \cite{csirmaz1997size} established the best known lower bound of $\Omega(n^2 / \log n)$ using information-theoretic inequalities. On the upper bound side, Liu and Vaikuntanathan \cite{liu2018breaking} constructed a scheme with share size $2^{0.994 n}$ for arbitrary access structures, while the current best bound is $2^{0.585 n}$, due to Applebaum et al.~\cite{applebaum2019secret, applebaum2020better}.

Benaloh and Leichter \cite{benaloh1988generalized} presented a general secret sharing construction that transforms any monotone access structure into a monotone Boolean function and then builds a perfect secret sharing scheme realizing that access structure.

Exploiting the equivalence between linear secret-sharing schemes and monotone span programs, as established by Beimel \cite{beimel1996secure}, Babai, G{\'a}l, and Wigderson \cite{babai1999superpolynomial} proved the first super-polynomial lower bound. Robere, Pitassi, Rossman, and Cook \cite{robere2016exponential} obtained an exponential lower bound on the size of monotone span programs for an explicit monotone function, which in turn implies a corresponding lower bound on the total share size of any linear secret-sharing scheme realizing that access structure.

In contrast to general access structure secret-sharing schemes, threshold schemes are relatively well understood. Shamir \cite{shamir1979share} introduced a threshold scheme based on polynomial evaluation and interpolation, which can be implemented in time $O(n \cdot \mathrm{polylog}(n))$ when the secret has length $O(\log n)$. Asmuth and Bloom \cite{asmuth1983modular} later proposed a threshold scheme based on the Chinese Remainder Theorem.

Bogdanov, Guo and Komargodski proved that for any $t < n$, a $(t, n)$-threshold SS scheme for one-bit secrets requires share size $\log(t + 1)$ \cite{bogdanov2016threshold}. As a consequence, the total share sizes must be $\Omega(n \log n)$ when $t = \Omega(n)$ for one-bit $(t, n)$-threshold SS.

One variant is the \emph{near-threshold} secret-sharing scheme, in which a $(\sigma, \rho)$-threshold scheme guarantees that any set of at most $\sigma n$ parties learns nothing about the secret, while any set of at least $\rho n$ parties can fully reconstruct it. Druk and Ishai \cite{druk2014linear} showed that for such schemes the shares can be computed by linear-size, logarithmic-depth circuits, building on the hash-function-based construction of Ishai et al.~\cite{ishai2008cryptography}. Cramer et al.~\cite{cramer2015linear} gave a construction supporting both linear-time sharing and reconstruction.

\subsection{Our results}

\textbf{Circuit model.} The computation model is that of \emph{unrestricted} arithmetic circuits over a finite field $\mathbb{F}$, as illustrated in Figure~\ref{fig:model}. We assume the secret is represented by an element of $\mathbb{F}$, and each share is also an element of $\mathbb{F}$. To realize the distribution of $n$ shares in a $(t,n)$-threshold SS scheme, the circuit has $t$ inputs and $n$ outputs, where one input corresponds to the secret and the remaining $t-1$ inputs are independent random elements of $\mathbb{F}$. The circuit is unrestricted in the sense that each gate can compute \emph{any} function and has unbounded fan-in. Accordingly, we measure circuit size by the number of wires rather than the number of gates.

Our first result gives a graph-theoretic condition that must be satisfied by any unrestricted arithmetic circuit computing the $n$ shares of a $(t,n)$-threshold SS scheme.

\begin{theorem}
\label{thm:main_condition}
An \emph{$(t, n)$-concentrator}, where $t \le n$, is a directed acyclic graph with $t$ inputs and $n$ outputs in which every set of $t$ outputs is connected to distinct $t$ inputs by vertex-disjoint paths.

Let $C$ be an unrestricted arithmetic circuit (arbitrary gates, unbounded fan-in) computing the $n$ shares of a $(t,n)$-threshold SS scheme. Assume $C$ has $t$ inputs, consisting of the secret $s$ and $t-1$ random field elements. Then:
\begin{itemize}
    \item $C$ is a $(t,n)$-concentrator.
    \item Removing $s$ yields a $(t-1,n)$-concentrator.
\end{itemize}
\end{theorem}

Theorem~\ref{thm:main_condition} is proved using Shannon-type information inequalities, which is later used to derive circuit lower bounds (see Theorem~\ref{thm:ss_ckt_lb} below). To the best of our knowledge, the use of Shannon-type information inequalities to prove circuit lower bounds is new.  
Our strategy is to use the well-known entropy characterization of threshold
secret sharing (namely, \eqref{equ:h_cor} and \eqref{equ:h_pri} below) to derive new
information inequalities—specifically, Theorem~\ref{thm:m} and
Theorem~\ref{thm:master_m1}—via Shannon-type information inequalities, and to
relate these inequalities to the connectivity properties of the circuit graph.

In contrast, Newman, Ragde, and Wigderson used \emph{graph entropy} to prove superlinear lower bounds on the formula size of certain Boolean functions. Beyond this, we are not aware of other techniques that use Shannon-type or non-Shannon-type information inequalities to establish circuit lower bounds.

Our second result establishes a reverse direction that complements Theorem~\ref{thm:main_condition}. We show that any graph satisfying the conditions of Theorem~\ref{thm:main_condition} can be transformed into an arithmetic circuit computing the shares of a threshold SS scheme. Together, Theorem~\ref{thm:main_condition} and Theorem~\ref{thm:main_construction} give a graph-theoretic characterization of unrestricted arithmetic circuits computing the shares of a threshold SS scheme.

\begin{theorem}
\label{thm:main_construction}
Let $G$ be a directed acyclic graph with $t$ inputs and $n$ outputs satisfying the conditions of Theorem~\ref{thm:main_condition}. We convert $G$ into an arithmetic circuit by replacing each non-input vertex with a weighted addition gate, where the coefficients are chosen independently and uniformly at random from $\mathbb{F}_q$. If $q$ is sufficiently large, then with high probability the resulting circuit computes the $n$ shares of a $(t,n)$-threshold SS scheme.
\end{theorem}

The proof technique originates from network coding~\cite{li2003linear} and has been applied in several other places. For example, Cheung, Kwok, and Lau used this idea to design fast algorithms for computing the rank of a matrix~\cite{cheng2017near}, and Drucker and Li used related ideas to construct circuits encoding error-correcting codes~\cite{Drucker2022}. In \cite{Drucker2022}, Drucker and Li introduced \emph{superconcentrator-induced codes} whose generator matrices are totally invertible, meaning that every square submatrix is invertible. Such codes can be used to realize threshold secret sharing schemes; in fact, as shown in Theorem~\ref{thm:main_construction}, a weaker condition already suffices.

The connection between threshold secret sharing and MDS codes was known \cite{karnin1983secret, blakley1995ideal}. We can therefore reformulate our main results in terms of the arithmetic circuit complexity of MDS codes.

Shah, Rashmi, and Ramchandran \cite{shah2013secure} studied the communication complexity of threshold secret sharing by establishing a necessary condition and a distinct sufficient condition, both expressed in graph-theoretic terms. Their necessary condition coincides with the first condition in Theorem~\ref{thm:main_condition}; their sufficient condition differs from ours.

As a consequence of the graph-theoretic characterization, we derive asymptotically tight lower bounds on the size of bounded-depth circuits that compute the shares of a threshold SS scheme.

\begin{theorem}
\label{thm:ss_ckt_lb}
Let \(\mathbb{F}_q\) be a field and let \(c \in (0,1)\) be a constant.
Let \(C : \mathbb{F}_q^{t} \to \mathbb{F}_q^{n}\), with \(t = cn\), be an
unrestricted arithmetic circuit of depth \(d\) (allowing arbitrary gates
and unbounded fan-in) that computes the shares of a \((t,n)\)-threshold
secret sharing scheme. The inputs consist of the secret together with
\(t-1\) random elements, and the outputs are the \(n\) shares.
Then the number of wires in \(C\) is at least $\Omega_{d, c}\!\left(\lambda_d(n)\cdot n\right)$.
\end{theorem}

Theorem~\ref{thm:ss_ckt_lb} is proved by combining Theorem~\ref{thm:main_condition}, the connectivity requirement satisfied by circuits encoding good error-correcting codes as shown by G{\'a}l et al.~\cite{gal2013tight} (straightforwardly extended from $\mathbb{F}_2$ to a large finite field), and the size bounds for densely regular graphs due to Pudl{\'a}k~\cite{Pud94}.

Our next result is a non-explicit construction of unbalanced superconcentrators. The construction follows classical techniques for balanced superconcentrators, such as those in \cite{Pinsker73, Val77, DDPW83}.

\begin{theorem}
\label{thm:main_sc}
For any $m \ge n$, there exists an $(m, n)$-superconcentrator with $O(m)$ edges and with depth $\alpha(m, n) + O(1)$, where $\alpha(m, n)$ is a version of the two-parameter inverse Ackermann function. 
\end{theorem}

Combining Theorems~\ref{thm:main_construction} and \ref{thm:main_sc}, we conclude that a $(t, n)$-threshold secret sharing scheme can be implemented by a linear-size circuit of depth $O(\alpha(t, n))$. For instance, depth~2 suffices when $n > t^{2.5}$, and depth~3 suffices when $n > t \log^{1.5} t$.

\section{Preliminaries}

\subsection{Entropy and information inequalities}

Let $X$ be a random variable taking values in a finite set, and let
$p(x) = \Pr[X = x]$. The entropy of $X$ is
\[
H(X) = - \sum_x p(x)\log p(x).
\]
For random variables $X$ and $Y$, the conditional entropy of $Y$ given $X$ is
\[
H(Y \mid X) = - \sum_{x,y} p(x,y)\log p(y \mid x),
\]
where $p(y \mid x) = \Pr[Y = y \mid X = x]$. It is well known that
\[
H(Y \mid X) = H(X,Y) - H(X)
\quad\text{and}\quad
H(Y \mid X) = \sum_x p(x) H(Y \mid X=x).
\]

The mutual information between $X$ and $Y$ is
\[
I(X;Y) = \sum_{x,y} p(x,y)\log \frac{p(x,y)}{p(x)p(y)},
\]
and satisfies $I(X;Y)=I(Y;X)$ and
\begin{equation}
I(X;Y) = H(X) - H(X \mid Y) \ge 0.
\end{equation}
The conditional mutual information of $X$ and $Y$ given $Z$ is
\[
I(X;Y \mid Z)
= \sum_{x,y,z} p(x,y,z)
\log \frac{p(x,y \mid z)}{p(x \mid z)p(y \mid z)},
\]
and satisfies
\begin{equation}
I(X;Y \mid Z) = H(X \mid Z) - H(X \mid Y,Z) \ge 0.
\end{equation}

A linear information inequality is called \emph{Shannon-type} if it can be obtained as a nonnegative linear combination of the basic Shannon inequalities, namely the nonnegativity of conditional entropy
$H(X \mid Y) \ge 0$
and conditional mutual information
$I(X;Y \mid Z) \ge 0$.

We refer to~\cite{yeung2002first} for background on entropy functions and information inequalities.

\subsection{Secret sharing scheme}

A \emph{secret sharing (SS) scheme} allows a dealer to distribute a secret among a group of participants such that
\begin{itemize}
	\item (Correctness) any authorized subset of participants can fully recover the secret, and
	\item (Privacy) any unauthorized subset of participants learns nothing about the secret.
\end{itemize}

One widely studied type of SS scheme is the \emph{$(t,n)$-threshold SS scheme}, in which any subset of $t$ participants (out of $n$ participants) can recover the secret, while any subset of at most $t-1$ participants learns nothing about it.

Fix a finite field $\mathbb{F}$. We represent the secret $s$ as an element of $\mathbb{F}$. (If the secret size is larger than the field size, the secret can be divided into smaller pieces, and the SS scheme can be applied to each piece separately.) Let $r_1, r_2, \dots$ be random elements of $\mathbb{F}$ used by the SS scheme, and let $R = \{r_1, r_2, \ldots\}$. A scheme is called a \emph{linear SS scheme} if each share is a linear combination of the secret $s$ and $r_1, r_2, \ldots$ over the field $\mathbb{F}$.

It is well known that the requirements of a secret sharing scheme can be characterized using entropy functions~\cite{beimel2011secret}. Let the secret be represented by a random variable $S$, and let the $n$ shares be represented by random variables $Y_1, Y_2, \ldots, Y_n$. The scheme is a $(t,n)$-threshold SS scheme if and only if the following conditions hold:
\begin{itemize}
	\item (Correctness) For every $T = \{i_1, \ldots, i_t\} \subseteq [n]$ of size $t$,
	\begin{equation}
	\label{equ:h_cor}
	H(S \cond Y_T) = 0,
	\end{equation}
	where $Y_T$ denotes the vector $(Y_{i_1}, \ldots, Y_{i_t})$.
	\item (Privacy) For every $T \subseteq [n]$ of size $t-1$,
	\begin{equation}
	\label{equ:h_pri}
	H(S \cond Y_T) = H(S).
	\end{equation}
\end{itemize}

\subsection{Arithmetic circuit model}

In the arithmetic circuit model, we assume
\begin{itemize}
\item The secret is an element of a finite field $\mathbb{F}$.

\item During the distribution phase, at most $\ell-1$ random field elements are used, denoted by $r_1, \ldots, r_{\ell-1}$. Let $R = {r_1, \ldots, r_{\ell-1}}$. (As shown later, any $(t,n)$-threshold SS scheme requires at least $t-1$ random elements.)

\item The $n$ shares are computed by an arithmetic circuit over $\mathbb{F}$ with $n$ outputs, denoted by $y_1, y_2, \ldots, y_n$.

\item Gates have unbounded fan-in. The size of the circuit is measured by the number of wires.

\item For the lower bounds proved in this work, we assume the arithmetic circuit is \emph{unrestricted}: a gate with fan-in $d$ may compute an arbitrary function from $\mathbb{F}^d$ to $\mathbb{F}$, where $d$ is unbounded.

\item For the upper bounds in this work, we assume that each gate is a \emph{weighted addition gate}. That is, a gate $g(x_1, \ldots, x_m)$ with inputs $x_1, \ldots, x_m$ computes
\[
g(x_1, \ldots, x_m) = \sum_{i=1}^m c_i x_i,
\]
where $c_1, \ldots, c_m \in \mathbb{F}$ are fixed coefficients. Note that a weighted addition gate with $m$ inputs can be realized by first applying $m$ multiplication gates (one per input) followed by a bounded fan-in addition gate, resulting in total size $O(m)$ and depth $2$.

\item In a linear SS scheme, the circuit computes a linear transformation $X \mapsto MX$, where $X \in \mathbb{F}^{\ell}$ and $M$ is an $n \times \ell$ matrix. (However, internal gates need not be linear.)

\item An arithmetic circuit is called \emph{linear} if every gate computes a linear function over $\mathbb{F}$. In particular, a linear arithmetic circuit computes a linear transformation.
\end{itemize}

\begin{figure}[h]
\centering
\includegraphics[scale=0.4]{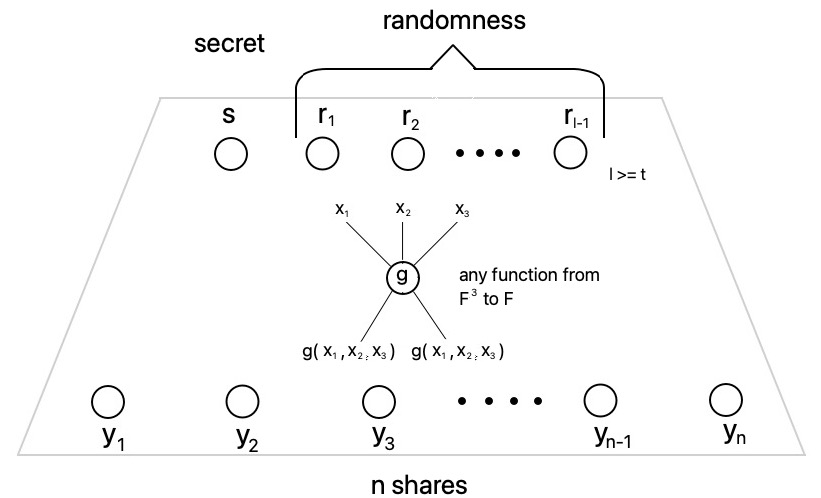}
\caption{unrestricted arithmetic circuit computing $n$ shares}
\label{fig:model}
\end{figure}

\subsection{Inverse Ackermann-type function}

Following Raz and Shpilka \cite{raz2001lower}, we define slowly-growing functions $\lambda_d(n)$. These are inverse Ackermann-type functions that are tailored for superconcentrators.

\begin{definition} For a function $f$, define $f^{(i)}$ to be the composition of $f$ with itself $i$ times, i.e., $f^{(i)} = \underbrace{f\circ f\circ \ldots \circ f}_{i \text{ times}}$. Thus, $f^{(1)} = f$.

For a function $f: \mathbb{N} \to \mathbb{N}$ such that $f(n) < n$ for all $n > 1$, define
\[
f^*(n) = \min\{ i : f^{(i)}(n) \le 1 \}.
\]
\end{definition}

\begin{proposition} Let $f : \mathbb{N} \to \mathbb{N}$ be any function such that $f(n) < n$ for all $n > 1$. Then we have
\[
f^*(n) \le f(n)-1.
\]
\end{proposition}
\begin{proof}
Consider
\[
f^{(1)}(n), \ldots, f^{(i)}(n) = 1,
\]
where $i = f^*(n)$.
Since $n > f(n) > f^{(1)}(n) > \ldots > f^{(i)}(n) = 1$, we have $i \le f(n) - 1$.
\end{proof}

\begin{definition} \cite{raz2001lower} Let
\begin{eqnarray*}
	\lambda_1(n) &=& \lfloor \sqrt{n} \rfloor, \\
	\lambda_2(n) &=& \lceil \log(n) \rceil, \\
	\lambda_d(n) & = & \lambda_{d-2}^*(n), \text{ for } d \ge 3
\end{eqnarray*}
\end{definition}

As $d$ increases, the functions $\lambda_d(n)$ grow extremely slowly. One can verify the following:
\[
\lambda_2(n) = \Theta(\log n), \qquad
\lambda_3(n) = \Theta(\log\log n), \qquad
\lambda_4(n) = \Theta(\log^* n), \qquad
\lambda_5(n) = \Theta(\log^* n).
\]

To analyze the size bounds of unbalanced superconcentrators, we require a two-parameter version of the inverse Ackermann function. The constants in the following definition are not optimized.

\begin{definition}
\label{def:inv_ack}
[Inverse Ackermann function] For any $m \ge n$, define
\begin{equation}
    \alpha(m, n) =
    \begin{cases}
      \min\{ d : \frac{m}{n} \ge \lambda_d(n) \}, & \text{if}\ m \ge 128 n, \\
      \min\{ d : \lambda_d(n) \le 4 \}, & \text{otherwise.}
    \end{cases}
\end{equation}
We denote $\alpha(n, n)$ simply by $\alpha(n)$, which recovers the standard one-parameter inverse Ackermann function.
\end{definition}

There exist multiple variants of the Ackermann function in the literature. To the best of our knowledge, all reasonable definitions of the inverse Ackermann function differ only by a multiplicative constant factor.

To show that $\alpha(m, n)$ is well-defined and to facilitate its use in the construction of unbalanced superconcentrators, we require the following properties. Proofs are deferred to Appendix~\ref{app:inverse_ack_properties}.

\begin{proposition}
\label{prop:u34d_ub}
	1. For any $n \ge 1$,
\[
\lambda_3(n) \le \log\log n + 2.
\]

2. For any $n \ge 1$,
\[
\lambda_4(n) \le 2 \log^*n.
\]

3. For any $d \ge 1$, for all $n \ge 4$,
\[
\lambda_d(n) \le n-2.
\]
\end{proposition}

\begin{proposition}
\label{prop:alpha_well_defined}
1. For any $d \ge 1$,
\[
\lambda_d(d) \le 4.
\]

2. For any $d \ge 1$, if $\lambda_d(n) \le C$, where $C \ge 128$, then
\[
\lambda_{d+2}(n)^2 \le C.
\]
\end{proposition}

\subsection{Superconcentrators and concentrators}

An \emph{$(m, n)$-network} is a directed acyclic graph with $m$ inputs and $n$ outputs.

\begin{definition} [$(m, n)$-superconcentrator \cite{Val77}]
An $(m, n)$-network is an \emph{$(m, n)$-superconcentrator} if for any subsets $X \subseteq I$ and $Y \subseteq O$ of equal size, there exist $|X|$ vertex-disjoint paths connecting $X$ to $Y$.
\end{definition}

In this definition, $m$ and $n$ need not be equal.
When $m = n$, tight bounds on the size of bounded-depth superconcentrator are known, achieved in a series of papers. 

\begin{table}[h!]
\begin{center}
\begin{tabular}{ |c|c|c| } 
\hline
Depth & Size  \\
\hline
2 & $\Theta(n \log^2 n / \log\log n)$ \cite{AP94, RT00}  \\ 
3 & $\Theta(n \log\log n)$ \cite{AP94} \\ 
$d \ge 4$ & $\Theta(n \lambda_d(n))$ \cite{DDPW83, Pud94}  \\ 
$\Theta(\alpha(n))$ & $\Theta(n)$ \cite{DDPW83}  \\ 
\hline
\end{tabular}
\caption{Superconcentrator size bounds.}
\label{table:sc_bounds}
\end{center}
\end{table}

Another relevant concept is \emph{concentrators}, which are critical building blocks for constructing superconcentrators.

\begin{definition}[Concentrator]
\label{def:concentrator}
An \emph{$(m, n, c)$-concentrator} is an $(m, n)$-network with the following property:  for any $c$  vertices chosen from the smaller of the input and output sets, there exist $c$ vertex-disjoint paths connecting them to  $c$ distinct vertices in the other set.

In particular:
\begin{itemize}
    \item If $m \ge n$, then for any subset of $c$ inputs, there exist $c$ vertex-disjoint paths connecting them to $c$ distinct outputs.
    \item If $m < n$, then for any subset of $c$ outputs, there exist $c$ vertex-disjoint paths connecting them to $c$ distinct inputs.
\end{itemize}

An $(m, n, n)$-concentrator or $(m, n, m)$-concentrator is called a \emph{full-capacity concentrator}, and is simply denoted as an $(m, n)$-concentrator.
\end{definition}

Every $(m, n)$-superconcentrator is an $(m, n)$-concentrator. Concentrators serve as building blocks in the construction of superconcentrators.

Using a standard probabilistic argument, one can prove 
\begin{lemma}
\label{lem:linear_concentrator}
\cite{pinsker1973complexity, DDPW83} 
For any integers $m, n, k$ satisfying $n \ge m$, $n \ge 1.1\,k$, and $m \ge k+1$, there exists a depth-1 $(m, n, k)$-concentrator of size
\[
O\Bigl(n \cdot \frac{\log(n/k)}{\log(m/k)}\Bigr).
\]
\end{lemma}

\section{Characterization via graph-theoretic properties}

\subsection{Necessity}

In this section, we use Shannon's information measures to show that the connectivity properties must hold for all circuits computing threshold secret-sharing schemes, linear or nonlinear.

The computation model is, again, an unrestricted arithmetic circuit as illustrated by Figure \ref{fig:model}, which computes the shares of some $(t, n)$-threshold SS scheme. The inputs are $s$ and $R = \{r_1, r_2, \ldots, r_{\ell-1} \}$, where $s \in \mathbb{F}$ is the secret, and $r_1, \ldots, r_{\ell-1} \in \mathbb{F}$ are independently and uniformly distributed over $\mathbb{F}$; the outputs are $y_1, \ldots, y_n$, representing $n$ shares.

Our goal is to prove that any circuit computing the shares of some $(t, n)$-threshold SS scheme must satisfy
\begin{itemize}
	\item For any subset of outputs $T \subseteq [n]$ of size $t-1$, there are $t-1$ vertex-disjoint paths connecting $R$ and $T$.
	\item For any subset of outputs $T \subseteq [n]$ of size $t$, there are $t$ vertex-disjoint paths connecting inputs (i.e., $R \cup \{s\}$) and $T$.
\end{itemize}

Our strategy is to formulate the \emph{connectivity requirements} as \emph{information inequalities}. We rely on two key observations: each gate carries at most $\log |\mathbb{F}|$ units of information; if random variables $Y$ can be written as a function of random variables $X$, then $H(Y) \le H(X)$.
So, it suffices to prove
\begin{itemize}
	\item For any subset of outputs $T \subseteq [n]$ of size $t-1$, 
	\begin{equation}
	\label{equ:goal_tm1}
		H(Y_T \cond S) \ge (t-1) H(S).
	\end{equation}
	\item For any subset of outputs $T \subseteq [n]$ of size $t$,
    \begin{equation}
    \label{equ:goal_t}
    	H(Y_T) \ge t H(S).
    \end{equation}
\end{itemize}
Given \eqref{equ:h_cor} and \eqref{equ:h_pri}, it turns out inequalities \eqref{equ:goal_tm1} and \eqref{equ:goal_t} can be proved using Shannon-type inequalities.
 
Before proving the information inequalities \eqref{equ:goal_tm1} and \eqref{equ:goal_t}, we first prove the following lemma.
 
\begin{lemma}
\label{lem:sum_nm1_nonnegative}
Let $Y_1, Y_2, \ldots, Y_n$ be random variables. Then
\[
\sum_{j \in [n]} H(Y_{[n] \setminus \{j\}}) - (n-1) H(Y_1, \ldots, Y_n) \ge 0. 
\]
\end{lemma}
\begin{proof} Write
\begin{align}
 & \sum_{j \in [n]} H(Y_{[n] \setminus \{j\}}) - (n-1) H(Y_1, \ldots, Y_n) \nonumber\\
 = & \sum_{j \in [n]} \left(H(Y_{[n] \setminus \{j\}}) - H(Y_1, \ldots, Y_n) \right) + H(Y_1, \ldots, Y_n). \label{equ:sum_j}
\end{align}
By the chain rule
\begin{align*}
 & H(Y_1, Y_2, \ldots, Y_n) \\
 = & H(Y_1) + H(Y_2 \cond Y_1) + H(Y_3 \cond Y_1, Y_2) + \ldots + H(Y_n \cond Y_1, Y_2, \ldots, Y_{n-1}) \\
 \ge & H(Y_1 \cond Y_{[n] \setminus \{1\}}) + H(Y_2 \cond Y_{[n] \setminus \{2\}}) + \ldots + H(Y_n \cond Y_{[n] \setminus \{n\}}),
\end{align*}
since conditioning reduces entropy.
Plugging it into \eqref{equ:sum_j}, we have
\begin{align*}
 & \sum_{j \in [n]} H(Y_{[n] \setminus \{j\}}) - (n-1) H(Y_1, \ldots, Y_n)\\
 \ge & \sum_{j \in [n]} \left(H(Y_{[n] \setminus \{j\}}) - H(Y_1, \ldots, Y_n) \right) + \sum_{j \in [n]} H(Y_j \cond Y_{[n] \setminus \{j\}}) \\
 = & \sum_{j \in [n]} \left(H(Y_{[n] \setminus \{j\}}) - H(Y_1, \ldots, Y_n) + H(Y_j \cond Y_{[n] \setminus \{j\}})\right) \\
 = & 0.
\end{align*}
This proves the lemma.
\end{proof}
 
 \begin{theorem}
 \label{thm:m}
 Let $S, Y_1, Y_2, \ldots, Y_n$ be random variables satisfying
 \begin{itemize}
 	\item $H(S \cond Y_T) = H(S)$ for any $T \subseteq [n]$ of size $t-1$, and
 	\item $H(S \cond Y_T) = 0$ for any $T \subseteq [n]$ of size $t$.
 \end{itemize}
 Then, we have
 \[
 H(Y_T) \ge t H(S)
 \]
 for any $T \subseteq [n]$ of size $t$.	
 \end{theorem}
 \begin{proof}
 Since the assumptions hold for every subset $T$ of size $t$, it suffices to prove the claim for $T = \{1,...,t\}$, that is, $H(Y_1, \ldots, Y_t) \ge t H(S)$. 

 We decompose $H(Y_1, \ldots, Y_t) - t H(S)$ into three nonnegative terms:
 \begin{align}
  & H(Y_1, \ldots, Y_t) - t H(S)  \nonumber\\
  = & \sum_{j \in [t]} \left( H(S, Y_{[t] \setminus \{j\}}) - H(Y_{[t] \setminus \{j\}}) - H(S) \right) - t\left( H(S, Y_{[t]}) - H(Y_{[t]})\right) \nonumber \\
  & + \sum_{j \in [t]} \left( H(S, Y_{[t]}) - H(S, Y_{[t] \setminus \{j\}}) \right)  + \sum_{j \in [t]} H(Y_{[t] \setminus \{j\}}) - (t-1)H(Y_{[t]}) \nonumber \\
  = & \sum_{j \in [t]} \left( H(S \cond Y_{[t] \setminus \{j\}}) - H(S) \right) - t H(S \cond Y_{[t]}) \label{term:m_t1} \\
  & + \sum_{j \in [t]} H(Y_j \cond S, Y_{[t] \setminus \{j\}}) \label{term:m_t2} \\
  & + \sum_{j \in [t]} H(Y_{[t] \setminus \{j\}}) - (t-1) H(Y_{[t]}). \label{term:m_t3}
 \end{align}
 The first term \eqref{term:m_t1} is zero by our conditions; the second term \eqref{term:m_t2} is  nonnegative; the third term \eqref{term:m_t3} is nonnegative by Lemma \ref{lem:sum_nm1_nonnegative}.
 Thus, we have $H(Y_1, \ldots, Y_t) - t H(S) \ge 0$, as desired.
 \end{proof}

\begin{theorem}
\label{thm:master_m1}
Let $S, Y_1, Y_2, \ldots, Y_n$ be random variables satisfying
 \begin{itemize}
 	\item $H(S \cond Y_T) = H(S)$ for any $T \subseteq [n]$ of size $t-1$, and
 	\item $H(S \cond Y_T) = 0$ for any $T \subseteq [n]$ of size $t$.
 \end{itemize}
Then, we have
\[
H(Y_T \cond S) \ge (t-1)H(S)
\]
for any $T \subseteq [n]$ of size $t-1$.
\end{theorem}
\begin{proof}
Since the assumptions and the claim are invariant under relabeling, we assume $T = \{1,2,\ldots,t-1\}$ without loss of generality. Write $H(Y_T \cond S) - (t-1) H(S)$ as
\begin{align}
 & H(Y_{[t-1]} \cond S) - (t-1) H(S) \nonumber \\
 = & H(S, Y_{[t-1]}) - t H(S) \nonumber \\
 = & \sum_{j \in [t-1]} \left( H(S, Y_{[t]}) - H(S, Y_{[t] \setminus \{j\}}) \right)
  + \sum_{j \in [t]} \left( H(S, Y_{[t]\setminus\{j\}}) - H(Y_{[t]\setminus\{j\}}) - H(S) \right) \nonumber \\
 & - (t-1) \left( H(S, Y_{[t]}) - H(Y_{[t]}) \right)  + \sum_{j \in [t]} H(Y_{[t] \setminus \{j\}}) - (t-1) H(Y_{[t]}) \nonumber \\
 = & \sum_{j \in [t-1]} H(Y_j \cond S, Y_{[t] \setminus \{j\}}) \label{term:m1_t1} \\
 & + \sum_{j \in [t]} \left( H(S \cond Y_{[t]\setminus\{j\}}) - H(S) \right) \label{term:m1_t2} \\
 & - (t-1) H(S \cond Y_{[t]}) \label{term:m1_t3} \\
 & + \sum_{j \in [t]} H(Y_{[t] \setminus \{j\}}) - (t-1) H(Y_{[t]}) \label{term:m1_t4},
\end{align}
where the term \eqref{term:m1_t1} is clearly nonnegative; terms \eqref{term:m1_t2} and \eqref{term:m1_t3} are zero due to our conditions; term \eqref{term:m1_t4} is nonnegative by Lemma \ref{lem:sum_nm1_nonnegative}. Thus, we have $H(Y_{[t-1]} \cond S) - (t-1) H(S) \ge 0$, as desired.
\end{proof}

\begin{theorem}[Menger's Theorem]
Let $G=(V,E)$ be an undirected graph and let $s,t \in V$ be distinct non-adjacent vertices.
The size of a minimum $s$--$t$ vertex cut equals the maximum number of pairwise
internally vertex-disjoint $s$--$t$ paths.
\end{theorem}

Now we are ready to prove our first theorem.

\begin{theorem} [Theorem \ref{thm:main_condition} restated]
\label{thm:graph_properties_main}
In the above model as illustrated by Figure \ref{fig:model}, if the circuit computes the shares of some $(t, n)$-threshold SS scheme, the following conditions are satisfied:
\begin{itemize}
	\item for any $T \subseteq [n]$ of size $t-1$, there are $t-1$ vertex-disjoint paths connecting $R$ and $T$;
	\item for any $T \subseteq [n]$ of size $t$, there are $t$ vertex-disjoint paths connecting inputs and $T$.
\end{itemize}
\end{theorem}
\begin{proof} 
Assume for contradiction that the first condition does not hold. That is, there exists a set $T \subseteq [n]$ of size $t-1$ such that there are at most $t-2$ vertex-disjoint paths from $R$ to $T$. By Menger's theorem, there exists a vertex set $U$ of size at most $t-2$ whose removal disconnects $R$ from $T$.

By the definition of the cut set $U$, we know that after setting $S$ to a constant, the outputs $Y_T$ can be written as functions in the gates in $U$. Since each gate value lies in $\mathbb{F}$ and hence carries at most $\log |\mathbb{F}|$ bits of entropy, we have
\[
H(Y_T \cond S = s) \le |U| \log |\mathbb{F}| \le (t-2) \log |\mathbb{F}|.
\]
Thus, 
\begin{align*}
	H(Y_T \cond S) = & \sum_s\Pr[S = s] H(Y_T \cond S = s) \\
	 \le & \sum_s\Pr[S = s] (t-2) \log |\mathbb{F}| \\
	 = & (t-2) \log |\mathbb{F}|.
\end{align*}
On the other hand, by Theorem~\ref{thm:master_m1}, we have
\[
H(Y_T \mid S) \ge (t-1) H(S) = (t-1) \log |\mathbb{F}|,
\]
since $S$ is uniformly distributed over $\mathbb{F}$. This is a contradiction.

The second condition can be proved similarly using Theorem \ref{thm:m}.
\end{proof}

In other words, the circuit, viewed as a graph, is an $(|R|+1, n, t)$-concentrator; moreover, after removing the input $s$ (along with its incident edges), the remaining graph is an $(|R|, n, t-1)$-concentrator.


\subsection{Sufficiency}

Given any $(t, n)$-superconcentrator $G$, or more generally any $(t, n)$-network satisfying the conditions of Theorem~\ref{thm:main_condition}, we construct a $(t, n)$-threshold secret-sharing scheme. 
The scheme is linear over a sufficiently large finite field $\mathbb{F}$, chosen such that
$|\mathbb{F}| \gg d \binom{n}{t},$
where $d$ denotes the \emph{depth} of the superconcentrator.

\begin{figure}[h]
\centering
\includegraphics[scale=0.4]{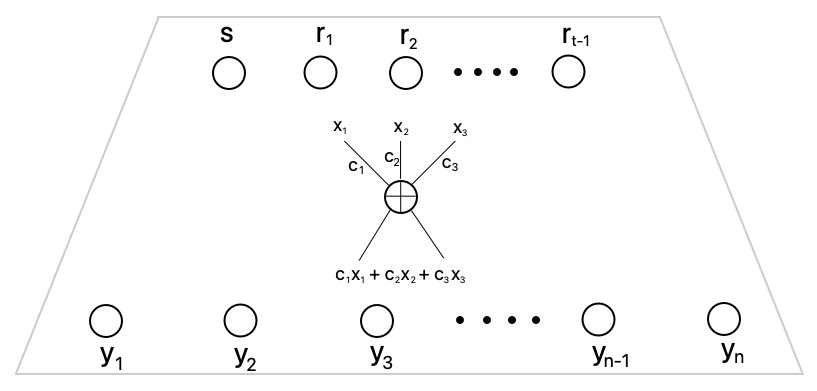}
\caption{linear arithmetic circuit realizing SS distribution}
\end{figure}

The SS scheme has the following 3 phases:

\textbf{Setup:} Convert $G$ into an arithmetic circuit $C$ over field $\mathbb{F}$ by
\begin{itemize}
\item replacing each vertex with an addition gate, and
\item for every edge $e$, choosing a coefficient $c_e \in \mathbb{F}$ \emph{uniformly at random}.
\end{itemize}
One can easily check this linear arithmetic circuit computes a linear transformation $x \mapsto Mx$, where
$M = (m_{ij})$ is a $n \times t$ matrix. Here
\[
m_{i,j} = \sum_{v_1 = x_j, v_2, \ldots, v_\ell = y_i} \prod_{k=1}^{\ell-1} c_{(v_k, v_{k+1})}
\]
where the sum ranges over all paths from input $x_j$ to output $y_i$.

\textbf{Sharing:} Assign the secret $s$ to input $x_1$, and choose $x_2, \ldots, x_t$ uniformly at random.
For every $i \in [n]$, send the $i$th output of the circuit to participant $P_i$.
In other words,
\[
\begin{pmatrix}
y_1 \\
y_2 \\
\vdots \\
y_n
\end{pmatrix}
=M
\begin{pmatrix}
s \\
r_2 \\
\vdots \\
r_t
\end{pmatrix},
\]
where $s \in \mathbb{F}$ is the secret, and $r_2, \ldots, r_t$ are uniformly random elements over $\mathbb{F}$.

\textbf{Reconstruction:} Consider any coalition of $t$ participants $T \subseteq [n]$ with shares
denoted by
\[
Y_T = M_T X = M_T \begin{pmatrix}
s \\
r_2 \\
\vdots \\
r_t
\end{pmatrix},
\]
where $X = (s, r_2, \ldots, r_t)^T$, and 
$M_T$ is the $t \times t$ submatrix of $M$ formed by rows indexed by $T$ and all columns.
Assuming $M_T$ is invertible (which will be proved), we have $X = M_T^{-1} Y_T$. The secret $s$ is then recovered as the first coordinate of $M_T^{-1} Y_T$.

\begin{lemma}
\label{lem:recover}
With probability at least $1 - \frac{d \binom{n}{t}}{|\mathbb{F}|}$, the secret can be recovered by any set of $t$ participants.
\end{lemma}
\begin{proof}
It suffices to show that, with probability at least $1 - \frac{d \binom{n}{t}}{|\mathbb{F}|}$, for all $T \subseteq [n]$ of size $t$, $\det(M_T) \neq 0$.
Because if $M_T$ is invertible, we can recover the secret by computing $X = M_T^{-1} Y_T$, where $Y_T$ denotes the shares received by the $t$ participants.

\begin{claim} For any $T \subseteq [n]$ of size $t$, we have
\[
\Pr[\det(M_T) = 0] \le \frac{d}{|\mathbb{F}|}.
\]
\end{claim}
\begin{proof} (of the Claim) Viewing the coefficients $c_e$ as the indeterminates, $\det(M_T)$ is a polynomial in $\mathbb{F}[\{c_e : e \in E(G)\}]$.

Note that there are $t$ vertex-disjoint paths from inputs to $T$. Setting the coefficients along these $t$ vertex-disjoint paths to $1$, and all other coefficients to $0$, the determinant evaluates to $\pm 1$. Hence, $\det(M_T)$ is a nonzero polynomial.

Observe that the polynomial $\det(M_T)$ has total degree $\le d$, where $d$ is the depth of the circuit. By Schwartz-Zippel Lemma, we have
\[
\Pr[\det(M_T) = 0] \le \frac{\deg(\det(M_T))}{|\mathbb{F}|} \le \frac{d}{|\mathbb{F}|}.
\]
\end{proof}

By the union bound over all $\binom{n}{t}$ choices of $T$, the probability that $\det(M_T) \neq 0$ for all $T$ is at least $1 - \frac{d \binom{n}{t}}{|\mathbb{F}|}$, as claimed.
\end{proof}

\begin{lemma}
\label{lem:zero_info}
With probability at least $1 - \frac{d \binom{n}{t-1}}{|\mathbb{F}|}$, any set of $t-1$ participants receives no information about the secret.
\end{lemma}
\begin{proof} 
Let $M_{T, \{2,3,\ldots,t\}}$ denote the $|T| \times (t-1)$ matrix indexed by rows $T$ and columns $2, 3, \ldots, t$, where $T \subseteq [n]$. For any $T \subseteq [n]$ of size $t-1$, we claim
\[
\Pr[\det M_{T, \{2,3,\ldots,t\}} = 0] \le \frac{d}{|\mathbb{F}|}.
\]
Viewing the coefficients $c_e$ as indeterminates over $\mathbb{F}$, $\det M_{T, \{2,3,\ldots,t\}}$ is a polynomial in $\{c_e : e \in E(G)\}$ of degree at most $d$. Since the circuit graph satisfies the second condition in Theorem \ref{thm:main_condition}, there exist $t-1$ vertex-disjoint paths connecting the inputs $x_2, \ldots, x_t$ to the outputs in $T$.
Setting the coefficients $c_e$ along these $t-1$ vertex-disjoint paths to $1$, and all other coefficients to $0$, the determinant evaluates to $\pm 1$, so $\det M_{T, \{2,3,\ldots,t\}}$ is a nonzero polynomial.
The claim follows from Schwartz-Zippel Lemma.

Taking a union bound over all $T \subseteq [n]$ of size $t-1$, we know with probability at least $1 - \frac{d{n \choose t-1}}{|\mathbb{F}|}$, $\det M_{T, \{2,3,\ldots,t\}} \not= 0$ for all $T$.

Consider any $t-1$ participants indexed by $T$, who receive the following vector in the reconstruction phase
\begin{eqnarray*}
y_T & = & M_{T, [t]} \begin{pmatrix}
s \\
r_2 \\
\vdots \\
r_t
\end{pmatrix} \\
& = & 
\begin{pmatrix}
M_{T,1} & M_{T, \{2, \ldots, t\}}
\end{pmatrix}
\begin{pmatrix}
s \\
R
\end{pmatrix} \\
& = &
M_{T, 1}s + M_{T, \{2, \ldots, t\}}R,
\end{eqnarray*}
where $R = \begin{pmatrix}
r_2 \\
\vdots \\
r_t
\end{pmatrix}$. 

Since $M_{T, \{2, \ldots, t\}}$ is of full rank, we know $M_{T, \{2, \ldots, t\}}R$ is uniformly distributed in $\mathbb{F}^{t-1}$ when $R \in \mathbb{F}^{t-1}$ is uniformly distributed.
Thus $M_{T, 1}s + M_{T, \{2, \ldots, t\}}R$ is uniformly distributed over $\mathbb{F}^{t-1}$, independent of $s$. Hence, any $t-1$ participants indexed by $T$ obtain zero information about the secret.
\end{proof}

Theorem \ref{thm:main_construction} immediately follows from Lemmas~\ref{lem:recover} and~\ref{lem:zero_info}.

The graph-theoretic condition required is slightly weaker than that of a superconcentrator; in fact, a concentrator suffices. For instance, consider a $(t-1, n)$-concentrator, add an additional input node $s$, and connect $s$ directly to all $n$ outputs. With this modification, the two connectivity conditions described above are satisfied.


\section{Consequences}

\subsection{Size lower bounds}

\begin{definition} \label{def:densely_regular} (Densely regular graph \cite{Pud94}) Let $G$ be a directed acyclic graph with $n$ inputs and $n$ outputs. Let $0 < \epsilon, \delta$ and $0 \le \mu \le 1$. We
say $G$ is \emph{$(\epsilon, \delta, \mu)$-densely regular} if for every $k \in [\mu n, n]$, there are probability distributions $\mathcal{X}$
and $\mathcal{Y}$ on $k$-element subsets of inputs and outputs respectively, such that for every $i \in [n]$,
\[
\Pr_{X \in \mathcal{X}}[i \in X] \le \frac{k}{\delta n}
\text{ and }
\Pr_{Y \in \mathcal{Y}}[i \in Y] \le \frac{k}{\delta n}
\]
and the expected number of vertex-disjoint paths from $X$ to $Y$ is at least $\epsilon k$ for randomly chosen $X \in \mathcal{X}$ and $Y \in \mathcal{Y}$.

Denote by $D(n, d, \epsilon, \delta, \eta)$ the minimal size of a $(\epsilon, \delta, \mu)$-densely regular layered directed acyclic graph with $n$ inputs and $n$ outputs and depth $d$. 
\end{definition}

The following result was proved for the case $\mathbb{F}_2$ (Lemma 3 in \cite{gal2013tight}); extending it to any finite field $\mathbb{F}_q$ is straightforward.

\begin{lemma}
\label{lem:vertex_disjoint_in_codes}
Let $\mathbb{F}_q$ be the finite field of size $q$. Let $\delta > 0$ be a constant. Let $C : \mathbb{F}_q^m \to \mathbb{F}_q^n$ be a code with minimum distance $\delta n$ and $G$ be an unrestricted arithmetic circuit (with arbitrary gates and unbounded fan-in computing $C$. For any $k \in \{1, \ldots, m\}$, and for any $k$-element subset $X$ of inputs of $G$, if we take uniformly at random a $k$-element subset $Y$ of outputs of $G$, then the expected number of vertex-disjoint paths from $X$ to $Y$ in $G$ is at least $\delta k$.
\end{lemma}

\begin{corollary}
\label{cor:codes_to_dr_graph}
(Corollary 15 in \cite{gal2013tight})
Let $0 < \rho, \delta < 1$ be constants and let $C:\mathbb{F}_q^{\rho n} \to \mathbb{F}_q^n$ be a circuit computing a code with relative distance $\delta$. If we extend the circuit by $(1-\rho)n$ dummy inputs, then its underlying graph is $(\rho\delta, \rho, \tfrac{1}{n})$-densely regular.
\end{corollary}

Corollary \ref{cor:codes_to_dr_graph} directly follows from Lemma \ref{lem:vertex_disjoint_in_codes}.

\begin{lemma}
\label{lem:concentrator_is_dr}
Let $G$ be a $(cn, n, cn)$-concentrator, where $c \in (0,1)$ is a constant. Then, after adding $(1-c)n$ dummy inputs, the graph is $(c(1-c), c, \tfrac{1}{n})$-densely regular.
\end{lemma}
\begin{proof} The proof proceeds in three steps. 
\begin{enumerate}
\item First, over a sufficiently large field, we transform the graph into a linear arithmetic circuit that computes a code of relative distance \(1-c\).
\item Second, by Lemma \ref{lem:vertex_disjoint_in_codes}, we know that any arithmetic circuit computing a code with relative distance \(1-c\) satisfies the connectivity requirement of Lemma~\ref{lem:vertex_disjoint_in_codes}.
\item Third, by applying Corollary~\ref{cor:codes_to_dr_graph}, we conclude that the graph is densely regular.
\end{enumerate}

Let $\mathbb{F}_q$ be a finite field, where $q$ is sufficiently large. We convert the $(cn, n, cn)$-concentrator graph $G$ into a linear arithmetic circuit, where each vertex replaced a weighted addition gate with random coefficients. Let $C:\mathbb{F}_q^{cn} \to \mathbb{F}_q^n$ be the linear mapping computed by the circuit, which computes a linear transformation $C(x) = Hx$, where $H$ is an $n \times cn$ matrix, and $x \in \mathbb{F}_q^{cn}$.

We claim that there exists a circuit \(C(x)=Hx\) such that, for every subset
\(S \subseteq [n]\) with \(|S| = cn\), the row submatrix \(H_S\) has full rank.
Indeed, for any fixed \(S\), the determinant \(\det(H_S)\) is a nonzero polynomial of degree at most the depth of the graph \(G\).
By the definition of a \((cn,n,cn)\)-concentrator, there exist \(cn\) vertex-disjoint paths connecting the \(cn\) inputs to the vertices in \(S\), which guarantees that \(\det(H_S)\) is not identically zero.
Therefore, by the Schwartz--Zippel lemma,
\[
\Pr[\det(H_S)=0] \le \frac{d}{q},
\]
where \(d\) denotes the depth of the graph \(G\).

Fix such a code \(C\). We claim that \(C\) has relative distance at least \(1-c\). Let \(x \in \mathbb{F}_q^{cn}\) be any nonzero vector. For any subset \(S \subseteq [n]\) with \(|S| = cn\), the submatrix \(C_S\) has full rank, and
hence \(C(x)_S = x C_S \neq \vec{0}\). Therefore, no nonzero codeword can be zero on more than \(cn-1\) coordinates, which implies that the Hamming weight of \(C(x)\) is at least \(n - cn + 1\). 

The second step follows from Lemma~\ref{lem:vertex_disjoint_in_codes}, and the third step follows from Corollary~\ref{cor:codes_to_dr_graph}.

\end{proof}

\begin{theorem} 
\label{thm:pudlak_dr_lb}
\cite{Pud94} Let $\epsilon, \delta > 0$ be constants. Then for every $n$, $\mu \in [1/n, 1]$, and $d \ge 3$, we have
\[
D(n, d, \epsilon, \delta, \mu) \ge \Omega_{d, \epsilon, \delta}(n \cdot \lambda_d(1/\mu)).
\]
\end{theorem}

\begin{theorem} [Theorem \ref{thm:ss_ckt_lb}]
Let \(\mathbb{F}_q\) be a field and let \(c \in (0,1)\) be a constant.
Let \(C : \mathbb{F}_q^{t} \to \mathbb{F}_q^{n}\), with \(t = cn\), be an
unrestricted arithmetic circuit of depth \(d\) (allowing arbitrary gates
and unbounded fan-in) that computes the shares of a \((t,n)\)-threshold
secret sharing scheme. The inputs consist of the secret together with
\(t-1\) random elements, and the outputs are the \(n\) shares.
Then the number of wires in \(C\) is at least $\Omega_{d, c}\!\left(\lambda_d(n)\cdot n\right)$.
\end{theorem}
\begin{proof} 
By Theorem~\ref{thm:main_condition}, the circuit \(C\), viewed as a graph \(G\),
is a \((cn,n,cn)\)-concentrator. By Lemma~\ref{lem:concentrator_is_dr},
the graph \(G\) is \((c(1-c),\,c,\,1/n)\)-densely regular. Finally,
Theorem~\ref{thm:pudlak_dr_lb} implies that \(G\) has size
\(\Omega_{d,c}\!\left(\lambda_d(n)\cdot n\right)\).
\end{proof}

\subsection{Size upper bounds}

Let $\scsize_d(m,n)$ denote the minimum size of a depth-$d$ superconcentrator with $m$ inputs and $n$ outputs, where $m$ and $n$ need not be equal ($m$ may be larger or smaller than $n$).

For computing SS schemes, we need unbalanced superconcentrators, where the number of inputs is the threshold value $t$, and the number of outputs is the number of participants $n$.

From Table \ref{table:sc_bounds}, we know there exists a linear-size $(n, n)$-superconcentrator of depth $O(\alpha(n))$. By removing some inputs (and the incident edges), we obtain an $O(n)$-size $(m, n)$-superconcentrator depth $O(\alpha(n))$, for any $m \le n$. Size $O(n)$ is clearly optimal (up to a multiplicative constant), since at least $n$ edges are required to connect the $n$ outputs. The question is, given $m, n$, can we achieve a depth smaller than $O(\alpha(n))$ for an $(m,n)$-superconcentrator?

\begin{definition} [Partial superconcentrator \cite{DDPW83}] An $(m, n)$-network is a $(p, q)$-partial superconcentrator if for any $S \subseteq [m]$ and $T \subseteq [n]$ with $|S| = |T|$ and $|S| \in [q,p]$, there exist $|S| - q$ vertex-disjoint paths connecting $S$ and $T$.
\end{definition}

Let $\scsize_d(m, n, p, q)$ denote the minimal size of an $(m, n)$-network of depth at most $d$ which is a $(p, q)$-partial superconcentrator.

\subsection{Depth 2}

In this subsection, we construct unbalanced superconcentrators of depth 2, which are used as building blocks for higher depth.

\begin{lemma} \label{lem:depth2_partial_sc} 
For any $r$, we have
\[
\scsize_2(n, m, n/r, 2n/(3r)) = O(m \log m).
\]
\end{lemma}
\begin{proof} 
As illustrated in Figure \ref{figure:sc_depth2}, we construct a depth-2 network with $n$ inputs, $m$ outputs, and a middle layer containing $4n/(3r)$ vertices. The construction consists of two layers:
\begin{itemize}
    \item Top layer: $(n, 4n/(3r), n/r)$-concentrator;
    \item Bottom layer: $(m, 4n/(3r), n/r)$-concentrator.
\end{itemize}

\begin{figure}[h]
\centering
\includegraphics[scale=0.4]{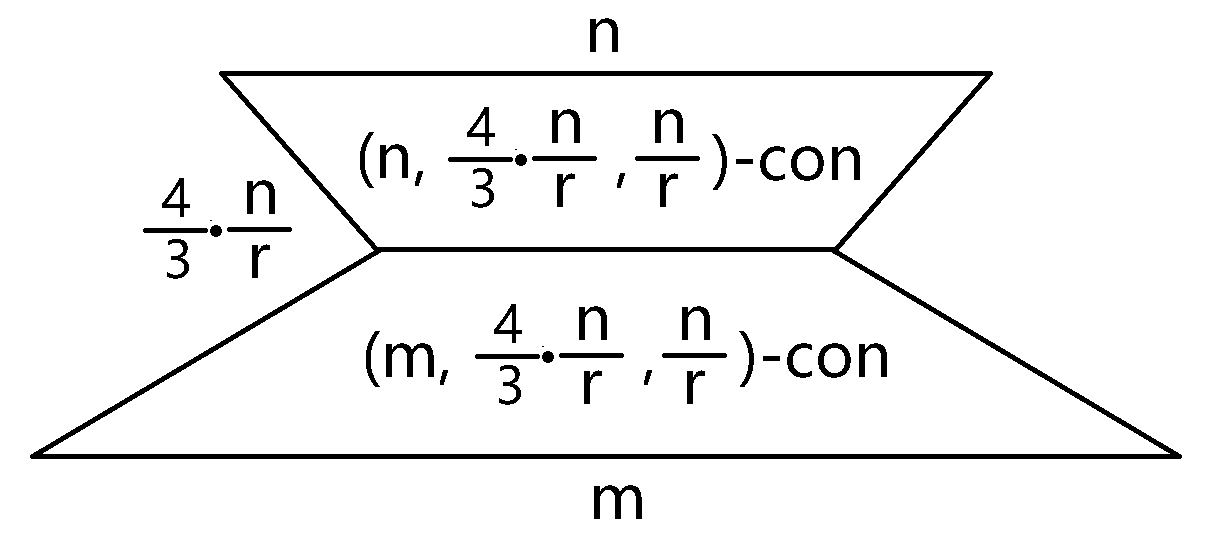}
\caption{Construction of depth-2 superconcentrator}
\label{figure:sc_depth2}
\end{figure}

Consider $S \subseteq [n]$ and $T \subseteq [m]$ of size $2n/(3r) + \Delta$, where $0 \le \Delta \le n/(3r)$. By the definition of the top-layer $(n,4n/(3r),n/r)$-concentrator, $S$ is connected to $|S|$ vertices in the middle layer, denoted by $S'$. Similarly, $T$ is connected to $|T|$ vertices in the middle layer, denoted by $T'$. Then
\[
|S' \cap T'| \ge |S'| + |T'| - 4n/(3r) = 2 \Delta.
\]

Thus, $S$ and $T$ are connected by $2\Delta$ vertex-disjoint paths (in fact, $\Delta$ paths suffice to satisfy the partial superconcentrator condition). Therefore, the $(n,m)$-network is a $(n/r, 2n/(3r))$-partial superconcentrator.
\end{proof}

By taking the union of $O(\log n)$ partial superconcentrators, we obtain the following upper bound on depth-2 superconcentrators.

\begin{lemma}
\label{lem:sc_depth2_size}
For any $n \le m$, we have
\[
\scsize_2(n, m) = O(m \log m \cdot \log n).
\]
\end{lemma}

\begin{proof} 
Let
\[
r = \left(\frac{3}{2}\right)^0, \left(\frac{3}{2}\right)^1, \ldots, \left(\frac{3}{2}\right)^{\ell}, \quad \text{where } \ell = \log_{3/2} n - 1.
\]
This gives a sequence of partial superconcentrators
\[
\left(n, \frac{2n}{3}\right), \left(\frac{2n}{3}, \frac{4n}{9}\right), \ldots, \left(\frac{n}{(3/2)^{\ell-1}}, \frac{n}{(3/2)^{\ell}}\right),
\]
each of size $O(m \log m)$ by Lemma \ref{lem:depth2_partial_sc}.

By combining these $\ell = O(\log n)$ partial superconcentrators and merging their inputs and outputs, we obtain an $(n, m)$-superconcentrator of size $O(m \log m \cdot \log n)$.
\end{proof}

When $m \ge n^{2+\Omega(1)}$, we prove there exists a depth-2 $(m, n)$-superconcentrator of linear size.

\begin{lemma}
\label{lem:sc_depth2_linear_size}
For any $\epsilon > 0$, if $m \ge n^{2+\epsilon}$,
\[
\scsize_2(m, n) = O\left(\frac{m}{\epsilon}\right).
\]
\end{lemma}

\begin{proof}
Construct a depth-2 $(m, n)$-network as illustrated in Figure \ref{figure:sc_depth2_linear}. Let the middle layer contain $\frac{m}{r}$ vertices, where
\[
r = \left(\frac{m}{n}\right)^{\frac{1}{1+\epsilon}}.
\]

\begin{itemize}
\item The top layer is a complete bipartite graph connecting all $n$ inputs to the middle layer vertices.
\item The bottom layer is a $(m, \frac{m}{r}, n)$-concentrator.
\end{itemize}

\begin{figure}[h]
\centering
\includegraphics[scale=0.6]{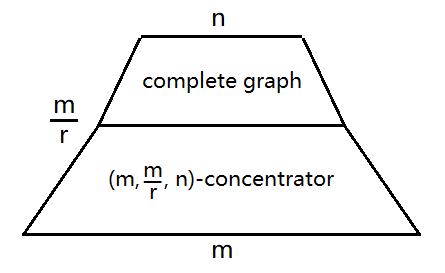}
\caption{Construction of linear-size depth-2 superconcentrator}
\label{figure:sc_depth2_linear}
\end{figure}

\textbf{Correctness:} For any subset of outputs $Y$ with $|Y| \le n$, the bottom-layer $(m, \frac{m}{r}, n)$-concentrator connects $Y$ to $|Y|$ middle-layer vertices. These middle-layer vertices are connected to any chosen $|Y|$ inputs via the complete bipartite top layer. Hence the network is indeed an $(m, n)$-superconcentrator.

\textbf{Size estimation:} The top-layer complete bipartite graph has size
\[
n \cdot \frac{m}{r} = n \cdot m \cdot \left(\frac{n}{m}\right)^{\frac{1}{1+\epsilon}} = m \cdot \left(\frac{n^{2+\epsilon}}{m}\right)^{\frac{1}{1+\epsilon}} \le m,
\]
since $m \ge n^{2+\epsilon}$.  

By Lemma \ref{lem:linear_concentrator}, the bottom-layer $(m, \frac{m}{r}, n)$-concentrator has size
\[
O\left(m \frac{\log(r^{1+\epsilon})}{\log(r^\epsilon)}\right) = O\left(\frac{m}{\epsilon}\right).
\]
Thus the total size is
$O(m) + O\left(\frac{m}{\epsilon}\right) = O\left(\frac{m}{\epsilon}\right).$
\end{proof}

\subsection{Depth 3}

When $m \ge n (\log n)^{2 + \Omega(1)}$, we prove there exists a depth-3 $(m, n)$-superconcentrator of linear size.

\begin{lemma}
\label{lem:d3_linear_sc}
For any $\epsilon > 0$, if $m \ge n (\log n)^{2 + \epsilon}$,
\[
\scsize_3(m, n) = O\left( \frac{m}{\epsilon} \right).
\]	
\end{lemma}

\begin{proof}
If $m \ge n^3$, by Lemma \ref{lem:sc_depth2_linear_size}, we have
\[
\scsize_3(m, n) \le \scsize_2(m, n) = O(m).
\]
Assume $m < n^3$ from now on.

\begin{figure}[h]
\centering
\includegraphics[scale=0.5]{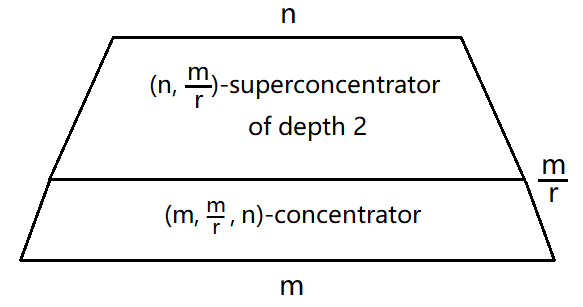}
\caption{Construction of linear-size depth-3 superconcentrator}
\label{figure:sc_depth3_linear}
\end{figure}

We construct an $(n, m)$-network consisting of two parts (Figure \ref{figure:sc_depth3_linear}):
\begin{itemize}
\item The top part is a depth-2 $(n, m/r)$-superconcentrator.
\item The bottom part is a $(m, m/r, n)$-concentrator, where
\[
r = \left( \frac{m}{n} \right)^{\frac{1}{1 + \epsilon/2}}.
\]
\end{itemize}

\textbf{Correctness:} Consider any subsets $X \subseteq [n]$ and $Y \subseteq [m]$ with $|X| = |Y|$. The bottom-layer $(m, m/r, n)$-concentrator connects $Y$ to $|Y|$ middle-layer vertices $Z$. The top-layer $(n, m/r)$-superconcentrator provides $|X|$ vertex-disjoint paths from $X$ to $Z$. Thus there are $|X|$ vertex-disjoint paths connecting $X$ and $Y$.

\textbf{Size estimation:} By Lemma \ref{lem:sc_depth2_size}, the top-layer depth-2 $(n, m/r)$-superconcentrator has size
\[
O\left(\frac{m}{r} \cdot \log \frac{m}{r} \cdot \log n \right) \le O\left(\frac{m}{r} (\log n)^2 \right) = O(m),
\]
since $\frac{m}{n} \ge (\log n)^{2+\epsilon}$ implies $(m/n)^{1/(1+\epsilon/2)} \ge (\log n)^2$.

By Lemma \ref{lem:linear_concentrator}, the bottom-layer $(m, m/r, n)$-concentrator has size
\[
O\left(m \frac{\log(r^{1+\epsilon/2})}{\log(r^{\epsilon/2})}\right) = O\left(\frac{m}{\epsilon}\right).
\]
Hence the total size is $O(m) + O\left(\frac{m}{\epsilon}\right) = O\left(\frac{m}{\epsilon}\right)$.
\end{proof}

\subsection{Higher depth}

As we have shown, when $m \ge n^{2+\epsilon}$ for some constant $\epsilon>0$, depth 2 suffices to achieve linear size; 
when $m \ge n (\log n)^{2+\epsilon}$, depth 3 suffices; 
and when $m$ is only ``slightly larger'' than $n$, higher depth is required.

In this subsection, we prove that for any $m \ge n$, depth $\alpha(m,n) + O(1)$ suffices, 
where $\alpha(m,n)$ denotes the two-parameter inverse Ackermann function (Definition \ref{def:inv_ack}).

\begin{lemma}
\label{lem:depth_d_sc_size}
For any depth $d \ge 3$, the size of a depth-$d$ $(m,n)$-superconcentrator satisfies
\[
\scsize_d(m, n) = O(m \, \lambda_d(n)).
\]
\end{lemma}

\begin{proof}
We consider two cases for $m$.

\textbf{Case 1:} $m \ge n^{2.5}$.  

By Lemma \ref{lem:sc_depth2_linear_size}, depth 2 suffices to construct an $(m,n)$-superconcentrator of size $O(m)$.  
Since increasing depth cannot increase size, for any $d \ge 2$ we have
\[
\scsize_d(m, n) \le \scsize_2(m, n) = O(m) = O(m \, \lambda_d(n)).
\]

\textbf{Case 2:} $m \le n^{2.5}$.  

Observe that, by monotonicity of $\lambda_d$, 
\[
\lambda_d(m) \le \lambda_d(n^{2.5}).
\]

If $d=3$, $\lambda_3(n^{2.5}) = O(\log\log n^{2.5}) = O(\log\log n) = O(\lambda_3(n))$.  
If $d \ge 4$, by the definition of $\lambda_d$ (iterated $\lambda_{d-2}$):
\[
\lambda_d(n^{2.5}) 
= \min \{ i : \lambda_{d-2}^{(i)}(n^{2.5}) \le 1 \} 
= \lambda_d(\lambda_{d-2}(n^{2.5})) + 1
\le \lambda_d(\lceil 2.5 \log n \rceil) + 1
\le \lambda_d(n) + 1
= O(\lambda_d(n)),
\]
where we used $\lambda_2(n) = \log n$ and the monotonicity of $\lambda_d$ in the last step.  

Hence, in all cases we have
\[
\lambda_d(m) \le O(\lambda_d(n)).
\]

Finally, from Table \ref{table:sc_bounds}, there exists a depth-$d$ \emph{superconcentrator} with $m$ inputs and $m$ outputs of size $O(m \, \lambda_d(m))$.  
By removing $m-n$ outputs and their incident edges, we obtain an $(m,n)$-superconcentrator.  
Thus,
\[
\scsize_d(m,n) \le \scsize_d(m,m) = O(m \, \lambda_d(m)) \le O(m \, \lambda_d(n)),
\]
as desired.
\end{proof}

\begin{theorem}
\label{thm:general_d_linear_sc}
For depth $d \ge 3$, if 
\[
m \ge n \, (\lambda_d(n))^{1+\epsilon}
\]
for some $\epsilon > 0$, then the minimal size of a depth-$(d+1)$ $(n,m)$-superconcentrator satisfies
\[
S_{d+1}(m, n) = O\left(\frac{m}{\epsilon}\right).
\]
\end{theorem}

\begin{proof}
We construct a depth-$(d+1)$ $(n,m)$-superconcentrator as follows.  

\medskip
\noindent\textbf{Construction:}  
\begin{itemize}
    \item Let $r = \left( \frac{m}{n} \right)^{\frac{1}{1+\epsilon}}$, so that $\frac{m}{r^{1+\epsilon}} = n$.  
    \item Place $\frac{m}{r}$ vertices on the second-to-last layer.  
    \item The first $d$ layers form an $(n, \frac{m}{r})$-superconcentrator. 
    \item The last layer is a $(m, \frac{m}{r}, n)$-concentrator connecting the second-to-last layer to the outputs.  
\end{itemize}

\medskip
\noindent\textbf{Verification:}  
Consider any subset of inputs $X \subseteq [n]$ and outputs $Y \subseteq [m]$ of equal size $|X| = |Y|$.  
By the definition of the $(m, \frac{m}{r}, n)$-concentrator, each output in $Y$ can be connected to a distinct vertex in the second-to-last layer via vertex-disjoint edges; denote these vertices by $Z$.  
By the definition of the $(n, \frac{m}{r})$-superconcentrator, each input in $X$ can be connected to a distinct vertex in $Z$ via vertex-disjoint paths in the first $d$ layers.  

Combining these two sets of vertex-disjoint paths, we obtain $|X| = |Y|$ vertex-disjoint paths connecting $X$ to $Y$.  
Hence, the constructed network is indeed an $(n,m)$-superconcentrator.

\medskip
\noindent\textbf{Size estimate:}  
By Lemma \ref{lem:depth_d_sc_size}, the first $d$ layers (depth-$d$ $(n, \frac{m}{r})$-superconcentrator) have size
\[
O\Big(\frac{m}{r} \cdot \lambda_d(n)\Big) = O\Big( m \cdot \frac{\lambda_d(n)}{(\frac{m}{n})^{1/(1+\epsilon)}} \Big).
\]
Using the assumption $m \ge n (\lambda_d(n))^{1+\epsilon}$, we have
\[
\frac{\lambda_d(n)}{(\frac{m}{n})^{1/(1+\epsilon)}} \le 1,
\]
so the first $d$ layers contribute $O(m)$ to the total size.

By Lemma \ref{lem:linear_concentrator}, the last layer ($(m, \frac{m}{r}, n)$-concentrator) has size
\[
O\left( m \cdot \frac{\log(r^{1+\epsilon})}{\log(r^\epsilon)} \right) = O\left( \frac{m}{\epsilon} \right).
\]

\medskip
\noindent\textbf{Total size:}  
Adding the contributions from all layers gives
$
S_{d+1}(m,n) = O\left( m + \frac{m}{\epsilon} \right) = O\left( \frac{m}{\epsilon} \right),
$
absorbing the $O(m)$ term into $O(m/\epsilon)$ since $\epsilon < 1$.  
\end{proof}

\begin{theorem}[\ref{thm:main_sc}]
For any $m \ge n$, if the depth satisfies 
$
d \ge \alpha(m,n) + 3,
$ 
then the minimal size of a depth-$d$ $(m,n)$-superconcentrator is linear in $m$, that is,
$
\scsize_d(m,n) = O(m),
$ 
where $\alpha(m,n)$ is the two-parameter inverse Ackermann function (Definition \ref{def:inv_ack}).
\end{theorem}

\begin{proof}
\noindent\textbf{Case 1: $m \ge 128\, n$.}  

By the definition of $\alpha(m,n)$, we have
$
\frac{m}{n} \ge \lambda_{\alpha(m,n)}(n),
$ 
and by Proposition \ref{prop:alpha_well_defined}, this implies
$
\frac{m}{n} \ge \lambda_{\alpha(m,n)+2}^{2}(n).
$
Then, applying Theorem \ref{thm:general_d_linear_sc} with depth $d = \alpha(m,n)+3$, we conclude
$
\scsize_{\alpha(m,n)+3}(m,n) = O(m).
$

\medskip
\noindent\textbf{Case 2: $n \le m < 128\, n$.}  If $\alpha(m, n) = 1$, we have $m \ge n \lfloor\sqrt{n} \rfloor$. By Lemma \ref{lem:d3_linear_sc}, we know $\scsize_4(m, n) \le \scsize_3(m, n) = O(m)$. If $\alpha(m, n) = 2$, we have $m \ge n \log n \ge n (2\log^* n)^2$. By Theorem \ref{thm:general_d_linear_sc}, $\scsize_5(m, n) = O(m)$.

From now on, we assume $\alpha(m, n) \ge 3$. By the definition of $\alpha(m, n)$, we have $\lambda_{\alpha(m, n)}(n) \le 4$. By Lemma \ref{lem:depth_d_sc_size}, we have
\[
\scsize_{\alpha(m,n)}(m, n) = O(m \lambda_{\alpha(m, n)}(n)) = O(m).
\]
\end{proof}

\section{Conclusion}

In this paper, we study the arithmetic circuit complexity of threshold secret sharing. We prove a graph-theoretic characterization of unrestricted arithmetic circuits that compute the shares of a threshold secret sharing scheme, when the underlying field is sufficiently large. As a result, we derive both lower and upper bounds.

\section*{Acknowledgements}
We thank the anonymous reviewers for their valuable comments and feedback.

\appendix

\section{Properties of inverse Ackermann function}

\label{app:inverse_ack_properties}

We include here the proofs of Proposition \ref{prop:u34d_ub} and \ref{prop:alpha_well_defined}.

\begin{proof} (of Proposition \ref{prop:u34d_ub})
	1. Recall that $\lambda_1(n) = \lfloor \sqrt{n} \rfloor$. We have
\begin{eqnarray*}
\lambda_3(n) & \le & \min\{d : n^{2^{-d}} \le 3 \} + 1 \\
& = & \lceil \log\left(\log n - \log 3\right) \rceil + 1 \\
& \le &  \log\left(\log n - \log 3\right) + 2 \\
& \le & \log\log n + 2.
\end{eqnarray*}

2. Recall that $\lambda_2(n) = \lceil \log n \rceil < \log n + 1$. Let $g(n) = \log n + 1$.
Observe that $g(g(n)) \le \log n$ for all $n \ge 8$. So we have
\begin{eqnarray*}
	\lambda_4(n) & \le & \min\{ d : g^{(d)}(n) \le 8 \} + 3 \\
	& \le & \min\{ d : \log^{(\lfloor d/2 \rfloor)}(n) \le 8 \} + 3.
\end{eqnarray*}
If $\log^{(\lfloor d/2 \rfloor)}(n) \le 8$, then $\log^{(\lfloor d/2 \rfloor + 2)}(n) \le \log 3 	\approx 1.58$. So, $\lfloor \frac{d}{2} \rfloor + 2 \le \log^* n$, which implies that $d \le 2(\log^*n - 2) + 1$.
 Thus $\lambda_4(n) \le 2 ((\log^*n) - 2) + 1 + 3 = 2 \log^*n$,
which holds when $n \ge 8$. When $n \in \{3, 4, \ldots, 7\}$, $\lambda_4(n) \le 2 \log^*n$ can be verified by direct computation.

3. We do induction on $d$.

When $d = 1$, $\lambda_d(n) = \lambda_1(n) = \lfloor \sqrt{n} \rfloor \le n - 2$ for all $n \ge 4$.

When $d = 2$, $\lambda_d(n) = \lambda_2(n) = \lceil \log n \rceil \le n - 2$ for all $n \ge 4$.

When $d \ge 3$, $\lambda_d(n) = \lambda^*_{d-2}(n) \le \lambda_{d-2}(n) \le n - 2$, where $\lambda^*_{d-2}(n) \le \lambda_{d-2}(n)$ is by induction hypothesis.
\end{proof}

\begin{proof} (of Proposition \ref{prop:alpha_well_defined})
1. When $d = 1$, $\lambda_1(1) = 1$; When $d = 2$, $\lambda_2(2) = 1$; When $d = 3$, $\lambda_3(3) = \lambda_1^*(3) = 1$.

We do induction on $d$, where $d \ge 2$. Assuming the conclusion is true for $d$, we prove it for $d + 2$.
\begin{eqnarray*}
\lambda_{d+2}(d+2) & = & \lambda_d^*(d+2) \\
& \le & \lambda_d^*(\lambda_d(d+2)) + 1 \\
& \le & \lambda_d^*(d) + 1,
\end{eqnarray*}
where $\lambda_d(d+2) \le d$ is by Proposition \ref{prop:u34d_ub}. So, $\lambda_{d+2}(d+2) \le  \lambda_d^*(d) + 1 \le \lambda_d^*(\lambda_d(d)) + 2 \le  \lambda^*_d(4) + 2 \le \max\{\lambda_2^*(4), \lambda_3^*(4)\} + 2 = 4$.

2. By the definition of $\lambda_{d+2}(n)$, we have $\lambda_{d+2}(n) = \lambda_d^*(n) \le \lambda_d^*(\lambda_d(n)) + 1 \le \lambda_d^*(C) + 1$.

When $d$ is odd, $\lambda_d^*(C) \le \lambda_1^*(C) \le \log\log C + 2$ by Proposition \ref{prop:u34d_ub}. When $C \ge 128$, we have $(\log\log C + 2)^2 \le C$.

When $d$ is even, $\lambda_d^*(C) \le \lambda_2^*(C) \le 2\log^* C$ by Proposition \ref{prop:u34d_ub}. When $C \ge 128$, we have $(2\log^* C + 1)^2 \le C$.
\end{proof}

\printbibliography

\end{document}